\documentclass[11pt]{article}

\usepackage{graphicx}
\usepackage{amsfonts}
\usepackage{amsmath}
\usepackage{amssymb}
\usepackage{amsthm}
\usepackage{ifthen}
\usepackage{epsfig}
\usepackage{xspace}
\usepackage{fullpage}

\newtheorem{thm}{Theorem}[section]

\newtheorem{clm}[thm]{Claim}
\newtheorem{lem}[thm]{Lemma}

\newtheorem{cor}[thm]{Corollary}

\def\E{\mathbb{E}}
\def\PRB{\mathbb{P}}
\newcommand\card[1]{\lvert #1 \rvert}
\newcommand{\cH}{{\cal H}\xspace}
\newcommand{\cE}{{\cal E}\xspace}

\begin{document}

\title{Random Preferential Attachment Hypergraphs\footnote{Supported in part by the Israel Science Foundation (grant 1549/13).}}
\author{
Chen Avin {\footnotemark[2] \footnotemark[4]}
\and
Zvi Lotker \footnotemark[2]
\and
David Peleg \footnotemark[3]
}

\def\thefootnote{\fnsymbol{footnote}}
\footnotetext[2]{
\noindent 
Department of Communication Systems Engineering, Ben Gurion University of the Negev, Beer-Sheva, Israel. E-mail:~{\tt \{avin,zvilo\}@cse.bgu.ac.il}.}
\footnotetext[4]{
\noindent 
Part of this work was done while the author was a long term visitor at ICERM, Brown University}
\footnotetext[3]{
\noindent
Department of Computer Science and Applied Mathematics, The Weizmann Institute of Science, Rehovot, Israel. E-mail:~{\tt david.peleg@weizmann.ac.il}.
Supported in part by the I-CORE program of the Israel PBC and ISF (grant 4/11).}

\date{\today}


\maketitle 
\thispagestyle{empty}

\begin{abstract}
The random graph model has recently been extended to a random 
preferential attachment graph model, in order to enable the study of
general asymptotic properties in network types that are better represented
by the preferential attachment evolution model than by the ordinary (uniform)
evolution lodel.
Analogously, this paper extends the random {\em hypergraph} model 
to a random {\em preferential attachment hypergraph} model.
We then analyze the degree distribution of random preferential attachment 
hypergraphs and show that they possess heavy tail degree distribution 
properties similar to those of random preferential attachment graphs. 
However, our results show that the exponent of the degree distribution is 
sensitive to whether one considers the structure as a hypergraph or as a graph.
\end{abstract}



\newpage

\section{Introduction}



Random structures have proved to be an extremely useful concept in many 
disciplines, including mathematics, physics, economics and communication 
systems. Examining the typical behavior of random instances of a structure
allows us to understand its fundamental properties.
The foundations of random graph theory were first laid in a seminal paper 
by Erd{\H{o}}s and R{\'e}nyi in the late 1950's \cite{erd6s1960evolution}. 
Subsequently,
several alternative models for random  structures, often suitable for other 
applications, were suggested. 
One of the most important alternative models is the 
\emph{preferential attachment} model \cite{barabasi1999emergence}, 
which was found particularly suitable for describing a variety of phenomena 
in nature, such as the ``rich get richer'' phenomena, which cannot be 
adequately simulated by the original Erd{\H{o}}s-R{\'e}nyi model.
It has been shown that the preferential attachment model captures some 
universal properties of real world social networks and complex systems, 
like heavy tail degree distribution and the ``small world'' phenomenon 
\cite{newman2010networks}. 


One limitation of graphs is that they only capture dyadic (or binary) 
relations.
In real life, however, many natural, physical and social phenomena involve
$k$-ry relations for $k>2$, and therefore can be more accurately represented 
by hypergraphs than by graphs. For example, collaborations among researchers, 
as manifested through joint coauthorships of scientific papers, 
may be better represented by hyperedges and not edges. 
Figure \ref{fig:example}(a) depicts the hypergraph representation 
for coauthorship relations on four papers:
paper 1 authored by $\{a,b,e,f\}$,
paper 2 authored by $\{a,c,d,g\}$, 
paper 3 authored by $\{b,c,d\}$ and 
paper 4 authored by $\{e,f\}$.
Likewise, wireless communication networks \cite{avin2014radio} 
or social relations captured by photos that appear in Facebook and other 
social media also form hyperedges \cite{zhang2010hypergraph}. 
Affiliation models \cite{lattanzi2009affiliation,newman2002random}, 
which are a popular model for social networks, are commonly interpreted 
as bipartite graphs, where in fact they may sometimes be represented 
more conveniently as hypergraphs. 
Figure \ref{fig:example}(b) presents the bipartite graph representation 
of the hypergraph $H$ of Figure \ref{fig:example}(a).
Sometimes, one can only access the \emph{observed graph} $G(H)$ of the original 
hypergraph $H$, that is, only the pairwise relation between players is available
(see  Figure \ref{fig:example}(c)). In some cases this structure may be 
sufficient for the application at hand, but in many other cases the hypergraph 
structure is more accurate and informative/


The study of hypergraphs, and in particular random hypergraph models, has its 
roots in a 1976 paper by Erd{\H{o}}s and Bollobas \cite{bollobas1976cliques}, 
which offers a model analogous to the Erd{\H{o}}s-R{\'e}nyi random graph model 
\cite{erd6s1960evolution}. Recently, several interesting properties regarding 
the evolution of random hypergraphs in this model were studied in 
\cite{cooper1996perfect,ellis2013regular,ghoshal2009random}.

The current paper is motivated by the observation that, just as in the 
random graph case, the random hypergraph model is not suitable for studying 
social networks. Our first contribution is in extending the concept of random 
preferential attachment graphs to 
{\em random preferential attachment hypergraphs}.
We believe the this natural model will turn out to be useful 
in the future study of social networks and other complex systems.

The main technical contribution is that we analyze the degree distribution of random preferential attachment hypergraphs and show that they possess heavy tail degree distribution properties, similar to those of random preferential attachment graphs. However, our results show that the exponent of the degree distribution is sensitive to whether one considers the structure as a hypergraph or as a graph.

As a reference point, we consider the random preferential attachment graph 
model of Chung and Lu \cite{chung2006complex}. 
In that model, starting from an initial graph $G_0$, 
at any time step there occurs an event of one of two possible types: 
(1) a {\em vertex-arrival event}, occuring with probability $p$, 
where a new \emph{vertex} joins the network and selects its neighbor among 
the existing vertices via preferential attachment, or
(2) an {\em edge-arrival event}, occuring with probability $1-p$,
where a new \emph{edge} joins the network and selects its two endpoints 
from among the existing vertices via preferential attachment.
It is shown in \cite{chung2006complex} that the degree distribution of the
random preferential attachment graph follows a power law,
i.e., the probability of a random vertex to be of degree $k$ is proportional 
to $k^{-\beta}$, with $\beta^G=2+\frac{p}{2-p}$. 
A similar result can be shown in a setting where, at each time step, 
$d$ edges join the graph instead of only one 
(in either a vertex event or an edge event)\cite{newman2010networks}. 
This result holds even if at each step a random number of edges join the 
network, so long as the expected number of new edges is $d$ 
and the variance is bounded.

The model proposed here extends Chung and Lu's \cite{chung2006complex} model 
to support hypergrpahs. 
That is, the process starts with an initial hypergrpah, and at each time step
a random hyperedge joins the network. With probabilty $p$ this new random 
hyperedge includes a new vertex, and with probabilty $1-p$ it does not. 
Our model allows the hyperedge sizes to be random (with some restrictions) 
and the members of each edge are selected randomly 
according to preferential attachment.

We show that the degree distribution of the resulting hypergraph 
(as well as the observed graph) follows a power law, but with an exponent 
$\beta^H=2+\frac{p}{\mu-p}$, where $\mu$ is the expected size of an hyperedge.

Our results indicate that one should be careful when studying an observed 
graph of a general $k$-ry relation. In particular, it makes a difference 
if the observed graph was generated by a graph or by a hypergraph 
evolution mechanism, since the two generate observed graphs 
with \emph{different} degree distributions.

\begin{figure}[t]
\centering
\begin{tabular}{ccc}
\includegraphics[width=35mm]{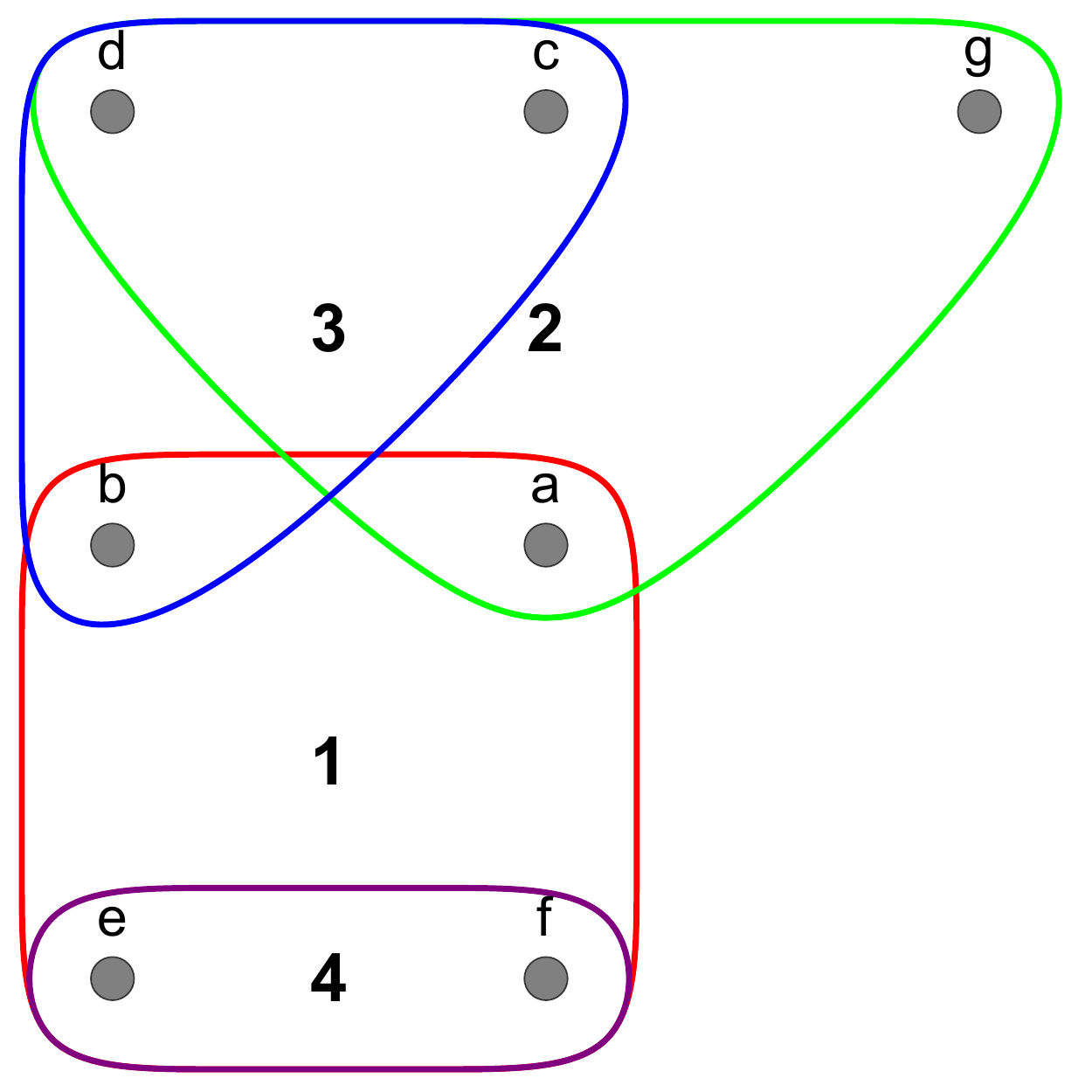}&
\includegraphics[width=40mm]{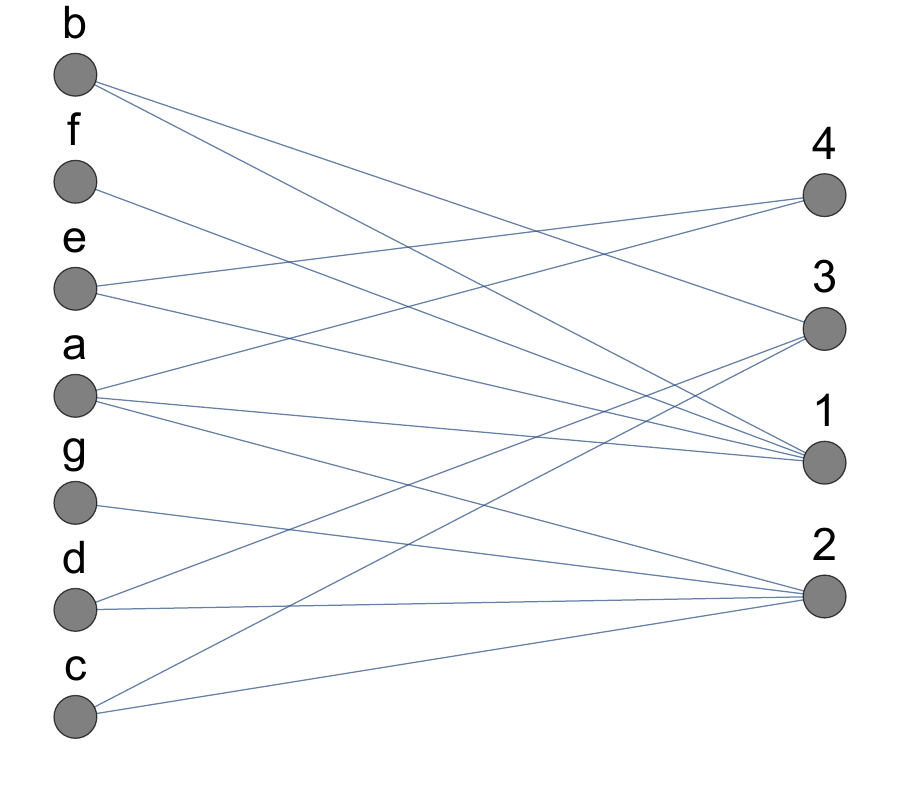}&
\includegraphics[width=50mm]{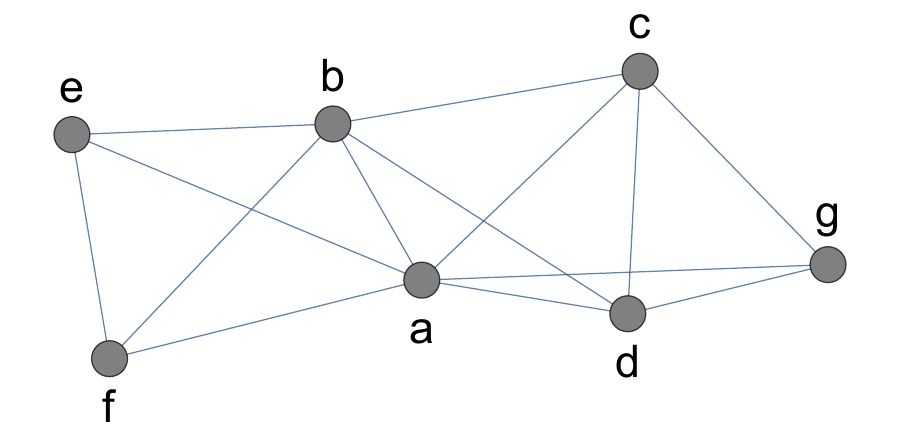}\\
(a) & (b) & (c)
\end{tabular}
\caption{
(a) A hypergraph $H$ with 7 vertices and 4 edges.
(b) A bipartite graph representation of $H$. 
(c) The observed graph $G(H)$.}
\label{fig:example}
\end{figure}

In the next sections we describe in more detail the preferential attachment 
model of a hypergraph, and then analyze the resulting degree distribution.

\section{Preliminaries}

Given a set $V$ and a natural $k>1$, let $V^{(k)}$ be the set of all 
unordered vectors (or multisets) of $k$ elements from $V$. 
A finite undirected \emph{graph} $G$ is an ordered pair $(V, E)$ where $V$ is 
a set of $n$ \emph{vertices} and $E \subseteq V^{(2)}$ is the set
of graph \emph{edges} (unordered pairs from $V$, including self loops).



A hypergraph $\mathcal{H}$ is an ordered pair $(V, \mathcal{E})$, where $V$ 
is a set of $n$ vertices and 
$\mathcal{E}  \subseteq \bigcup_{i=2}^n V^{(i)}$ is a set of 
{\em hyperedges} connecting the vertices
(including self loops). 
The \emph{rank} $r(\cH)$ of a hypergraph $\cH$ is the maximum cardinality 
of any of the hyperedges in the hypergraph. When all hyperedges have the same 
cardinality $k$, the hypergraph is said to be \emph{$k$-uniform}. 
A graph is thus simply a 2-uniform hypergraph.
The \emph{degree} of a hyperedge $e \in \cE$ is defined to be 
$\delta(e) = \card{e}$.
The set of all hyperedges that contain the vertex $v$
is denoted $\mathcal{E}(v)=\{e\in \mathcal{E} \mid v \in e\}$. 
The \emph{degree} $d(v)$ of a vertex $v$ is the number of hyperedges 
in $\cE(v)$, i.e., $d(v) = \card{\cE(v)}$.
$\mathcal{H}$ is $d$-\emph{regular} if every vertex has degree $d$.

In the classical preferential attachment graph model 
\cite{barabasi1999emergence}, the evolution process starts with an arbitrary 
finite initial network $G_0$, which is usually set to a single vertex 
with a self loop. Then this initial network evolves in time, with $G_t$ 
denoting the network after time step $t$. 
In every time step $t$ a new vertex $v$ enters the network. 
On arrival, the vertex $v$ attaches itself to an existing vertex $u$ chosen
at random with probability proportional to $u$'s degree at time $t$, i.e.,
$$\PRB\left[u \text{ is chosen}\right] ~=~ 
\frac{d_{t}(u)}{\sum_{w\in G_t}d_{t}(w)}~,$$
where $d_t(x)$ is the degree of vertex $x$ at time $t$.

\section{The nonuniform preferential attachment hypergraph model}

Similar to the classical preferential attachment graph model \cite{chung2006complex}, the evolution of the hypergraph occurs along a discrete time axis, with one event occurring at each time step.
We consider two types of possible events on the hypergraph at time $t$:
(1) a \emph{vertex arrival} event, which involves adding a new vertex 
along with a new hyperedge, and a \emph{hyperedge arrival} event, 
where a new hyperedge is added.

We consider a nonuniform, random hypergraph where self loops
(i.e., multiple appearance of a vertex in a hyperedge) are allowed. 
We consider self loops as contributing 1 to the vertex degree.
Similar to \cite{chung2006complex}, our preferential attachment model, 
$H(H_0, p, Y)$, has three parameters:
\begin{itemize}
\item A probability $0 < p \le 1$ for vertex arrival events.
\item An initial hypergraph $H_0$ given at time 0.
\item A sequence of random independent integer variables 
$Y = (Y_0,Y_1,Y_2, \dots)$, for $Y_i\ge 2$, which determine the cardinality
of the new hyperedge arriving at time $t$.
\end{itemize}

The process by which the random hypergraph $H(H_0, p, Y)$ grows in time
is as follows.

\begin{itemize}
\item We start with the initial hypergraph $H_0$ at time 0.
\item At time $t>0$, the graph $H_t$ is formed from $H_{t-1}$ 
in the following way:
\begin{itemize}
\item
Randomly draw a bit $b$ with probability $p$ for $b=0$.
\item If $b=0$, then add a new vertex $u$ to $V$, select $Y_t-1$ vertices
from $H_{t-1}$ (possibly with repetitions) independently in proportion
to their degrees in $H_{t-1}$, and form a new hyperedge $e$ that includes
$u$ and the $Y_t-1$ selected vertices\footnote{note that as the hypergraph 
gets larger, the probability of adding a self-loop is vanishing.}.
\item Else, select $Y_t$ vertices from $H_{t-1}$ (possibly with repetitions)
independently in proportion to their degrees in $H_{t-1}$, and
form a new hyperedge $e$ that includes the $Y_t$ selected vertices.
\end{itemize}
\end{itemize}

Hereafter, we consider an initial $H_0$ consisting of a single hyperedge
of cardinality $Y_0$ over a single vertex (recall that self-loops 
are considered as contributing 1 to the vertex degree).

\section{Degree Distribution Analysis}

To ensure convergence of the degree distribution we first need to set 
some conditions on the distribution of the hyperedge cardinalities. 
These are somewhat mild conditions that seems to agree with real data 
(see Fig. \ref{fig:dblp-edges-power} in Section \ref{sec:dblp}).
Let $Y_t$ be independent (not necessarily identical) random variables with
constant expectation $\E[Y_t]=\mu$ and bounded support
s.t. $2 < Y_t < t^{\frac{1}{3}}$ \footnote{The exponent $\frac{1}{3}$ is chosen somewhat arbitrarily; the result can be extended to any constant $0 \le \alpha < \frac{1}{2}$.}.
Under these conditions we can show the following.

\begin{thm}
The degree distribution of a hypergraph $H(H_0, p, Y)$ where $\E[Y_t]=\mu$
follows a power law with $\beta=2+p/(\mu-p)$.
\end{thm}

\begin{proof}
We start with properties of $Y_t$. 
Let $S_t= \sum_1^t Y_t$, so $\E[S_t]= \mu t$ and $S_t < t^{\frac{4}{3}}$.
The deviation of $S_t$ from its expected value can be bounded.

\begin{lem}
$\PRB \left [\card{S_t - \E[S_t]}) \ge t^{\frac{2}{3}}\sqrt{2\log t}\right] ~<~ O(1/t^4)$.
\end{lem}
\begin{proof}
By Hoeffding's inequality \cite{Hoeffding1963Probability}, 
assuming the random variable $Y_i$ satisfies $\PRB[Y_i \in [a_i, b_i]]=1$ 
for some reals $a_i$ and $b_i$,
$$\PRB \left[\card{S_t - \E[S_t]}) \ge x \right] ~\le~
2 \exp \left(-\frac{2x^2}{\sum_{i=1}^t (b_i-a_i)^2} \right)~.$$
Taking $x=t^{\frac{2}{3}}\sqrt{2\log t}$ and noting that 
$(b_i-a_i)^2 < t^\frac{2}{3}$ and $\sum_{i=1}^t (b_i-a_i)^2 < t^\frac{4}{3}$
yields the result.
\end{proof}

To bound the degree distribution of a non-uniform random hypergraph
we closely follow Chung and Lu's analysis on preferential attachment graphs
\cite{chung2006complex}.
Let $m_{k,t}$ denote the number of vertices of degree $k$ at time $t$. Note that
$m_{1,0} =0$ and $m_{0,t} = 0$.
We derive the recurrence formula for the expected value $\E[m_{k,t}]$.
The main observation here is that a vertex has degree $k$ at time $t$ if either
it had degree $k$ at time $t-1$ and was not selected into a hyperedge at time $t$,
or it had degree $k-1$ at time $t-1$ and was selected into a hyperedge at time $t$.
Letting $\mathcal{F}_t$ be the $\sigma$-algebra associated with
the probability space at time $t$, we have for any $t>0$ and $k>0$:
\begin{eqnarray*}
\E[m_{k,t} \vert \mathcal{F}_{t-1}] &=&
m_{k,t-1} \left(p\E_{Y_t}\left[\left(1-\frac{k}{S_{t-1}}\right)^{Y_t-1}\right]
 + (1-p)\E_{Y_t}\left[\left(1-\frac{k}{S_{t-1}}\right)^{Y_t}\right] \right)
\\
&& + m_{k-1,t-1}\left(p\E_{Y_t}\left[\left(1-(1-\frac{k-1}{S_t}\right)^{Y_t-1})\right]\right.
\\
&& + \left. (1-p)\E_{Y_t}\left[\left(1-\left(1-\frac{k-1}{S_t}\right)^{Y_t}\right)\right] \right),
\end{eqnarray*}
hence
\begin{eqnarray*}
\E[m_{k,t} \vert \mathcal{F}_{t-1}] &=&
m_{k,t-1} \left(p-\E_{Y_t}\left[\frac{(Y_t-1)kp}{S_t}\right]
 - O\left(\E_{Y_t}\left[\left(\frac{Y_tkp}{S_t}\right)^2\right]\right) \right.
\\
&& ~~~~~~~~~~~~ \left. + 1-p-\E_{Y_t}\left[\frac{(1-p)Y_tk}{S_t}\right]
 - O\left(\E_{Y_t}\left[\left(\frac{(1-p)Y_tk}{S_t}\right)^2\right]\right)\right)
\\
&& + m_{k-1,t-1} \left(\E_{Y_t}\left[\frac{(Y_t-1)p(k-1)}{S_t}\right]
- O\left(\E_{Y_t}\left[\left(\frac{Y_tpk}{S_t}\right)^2\right]\right) \right.
\\
&& ~~~~~~~~~~~~~~~~ \left. + \E_{Y_t}\left[\frac{(1-p)Y_t(k-1)}{S_t}\right]
 - O\left(\E_{Y_t}\left[\left(\frac{(1-p)k}{S_t}\right)^2\right]\right)\right)
\end{eqnarray*}
or
$$\E[m_{k,t} \vert \mathcal{F}_{t-1}] ~=~
m_{k,t-1} \left(1-\frac{(\mu-p)k}{S_t}+ O\left(\frac{k^2}{S_t^2}\right) \right)
+ m_{k-1,t-1} \left(\frac{(\mu-p)(k-1)}{S_t} 
+ O\left(\frac{(k-1)^2}{S_t^2}\right)\right)~.$$
Using the bound on $S_t$ we can find the expectation $\E[m_{k,t}]$.
\begin{eqnarray*}
\E[m_{k,t}] &=& (1-1/t^4)\left(\E[m_{k,t-1}]
\left(1-\frac{(\mu-p)k}{\mu t \pm t^{\frac{2}{3}}\sqrt{2\log t}}
+ O\left(\frac{k^2}{t^2}\right)\right)\right. \\
&& \;\;\;\; \left. + \E[m_{k-1,t-1}] \left(\frac{(\mu-p)(k-1)}{\mu t \pm t^{\frac{2}{3}}\sqrt{2\log t}} + O\left(\frac{k^2}{t^2}\right) \right) \right)  + \frac{1}{t^4} \cdot t^{4/3} \\
&=& \E[m_{k,t-1}] \left(1-\frac{(\mu-p)k}{\mu t \pm t^{\frac{2}{3}}\sqrt{2\log t}}
+ O\left(\frac{k^2}{t^2}\right) \right) \\
&& \;\;\;\; +  \E[m_{k-1,t-1}] \left(\frac{(\mu-p)(k-1)}{\mu t \pm t^{\frac{2}{3}}\sqrt{2\log t}} + O\left(\frac{k^2}{t^2}\right) \right) + O(1/t^2)~.
\end{eqnarray*}
For $t>0$ and the special case of $k=1$ we have
$$\E[m_{1,t} \vert \mathcal{F}_{t}] ~=~
m_{1,t-1}\left(1-\frac{(\mu-p)k}{S_t}+ O(\frac{k^2}{S_t^2}) \right) + p~,$$
thus
$$\E[m_{1,t}] ~=~ \E[m_{1,t-1}] \left(1-\frac{(\mu-p)k}{\mu t \pm t^{\frac{2}{3}}\sqrt{2\log t}} + O(\frac{k^2}{t^2}) \right) + p + O(1/t^2)~.$$
We use the following lemma of \cite{chung2006complex}.
\begin{lem}~\cite{chung2006complex}
\label{lem:magic}
Let $(a_t),(b_t),(c_t)$ be three sequences such that
$a_{t+1} = \left(1-\frac{b_t}{t}\right)a_t + c_t$, 
$\lim_{t\rightarrow\infty}b_t=b>0$ and 
$\lim_{t\rightarrow\infty}c_t = c$. Then
$\lim_{t\rightarrow \infty} (a_t/t)$ exists and equals $c/(1+b)$.
\end{lem}

We show by induction that $\lim_{t\rightarrow \infty} \E[m_{k,t}]/t$ exists and has a limit $M_k$ for each $k$.
For $k=1$, apply Lemma \ref{lem:magic} with
$$b_t = \frac{\mu-p}{\mu \pm t^{\frac{2}{3}}\sqrt{2\log t}/t} + O(k^2/t) 
~~\text{ and }~~
c_t = p + O(1/t^2)$$
and hence
$$\lim_{t\rightarrow\infty}b_t = \frac{\mu-p}{\mu} ~~\text{ and }~~
\lim_{t\rightarrow\infty} c_t = p,$$
to get
$$M_1 ~=~ \lim\limits_{t\rightarrow \infty} \frac{\E[m_{1,t}]}{t}
~=~ \frac{\mu p}{2\mu-p}~.$$
We now assume that $lim_{t\rightarrow\infty} \E[m_{k-1,t}]/t$ exists
and apply Lemma \ref{lem:magic} again with
$$b_t ~=~ \frac{(\mu-p)k}{\mu \pm t^{\frac{2}{3}}\sqrt{2\log t}/t} + O\left(\frac{k^2}{t}\right)$$
and
$$c_t ~=~  \frac{\E[m_{k-1,t-1}]}{t} \left(\frac{(\mu-p)(k-1)}{\mu \pm t^{\frac{2}{3}}\sqrt{2\log t}/t} + O\left(\frac{k^2}{t}\right) \right) 
+ O\left(\frac{1}{t^2}\right)~.$$
Then
$$\lim_{t\rightarrow\infty}b_t = b= \frac{(\mu-p)k}{\mu}
~~\text{ and }~~
\lim_{t\rightarrow\infty}c_t = c = M_{k-1}  (\mu-p)(k-1)/ \mu~,$$
and by Lemma \ref{lem:magic} we get that 
$lim_{t\rightarrow\infty}\E[m_{k,t}]/t$ exists and satisfies
\begin{align}
\label{Eq:recNon}
M_k ~=~ M_{k-1} \frac{(\mu-p)(k-1)}{\mu(1 + k(\mu-p)/\mu)}
~=~ M_{k-1} \frac{(k-1)}{k+ \frac{\mu}{\mu-p}}~.
\end{align}
Recall that a power law distribution has the property that
$M_k\propto k^{-\beta}$ for large $k$.

Now if $M_k\propto k^{-\beta}$, then
$$\frac{M_k}{M_{k-1}} ~=~ \frac{k^{-\beta}}{(k-1)^{-\beta}}
~=~ \left(1-\frac{1}{k}\right)^\beta
~=~ 1-\frac{\beta}{k}+O\left(\frac{1}{k^2}\right)~.$$
By Eq. (\ref{Eq:recNon}),
$$\frac{M_k}{M_{k-1}} ~=~ \frac{k-1}{k + \frac{\mu}{\mu-p}}
~=~ 1- \frac{1+\frac{\mu}{\mu-p}}{k+\frac{\mu}{\mu-p}}
~=~ 1-\frac{1+\frac{\mu}{\mu-p}}{k} + O\left(\frac{1}{k^2}\right)~,$$
so the exponent $\beta$ of the power law satisfies
$$\beta = 1 +\frac{\mu}{\mu-p} = 2 + \frac{p}{\mu-p}~.$$
\end{proof}

A special case of $H(H_0, p, Y_t)$ is when $Y_t$ is the constant function $d$ 
and the hypergraph becomes a $d$-uniform hypergraph denoted as $H(H_0, p, d)$.

\begin{cor}
The degree distribution of a $d$-uniform hypergraph $H(H_0, p, d)$
follows a power law with $\beta=2+p/(d-p)$.
\end{cor}

Figure \ref{fig:exponent} illustrates the difference in exponents $\beta$ 
between preferential attachment graphs (i.e., 2-uniform hypergraphs) and 
3-uniform hypergraphs as a function of $p$.

In many cases one can only observe the graph $G[H]$ that results of the 
underlying hypergrph $H$. That is, the set of vertices of $G(H)$ is identical 
to the set of vertices of $H$ and for every hyperedge $e \in H$ 
we create edges in $G(H)$ to form a clique between all the vertices in $e$. 
Now we can prove the following.

 \begin{clm}
The degree distribution of the observed graph $G(H(H_0, p, d))$ that results 
from a $d$-uniform hypergraph follows a power law with $\beta=2+p/(d-p)$.
\end{clm}

\begin{figure}[htb]
\centering
\includegraphics[width=.7\columnwidth]{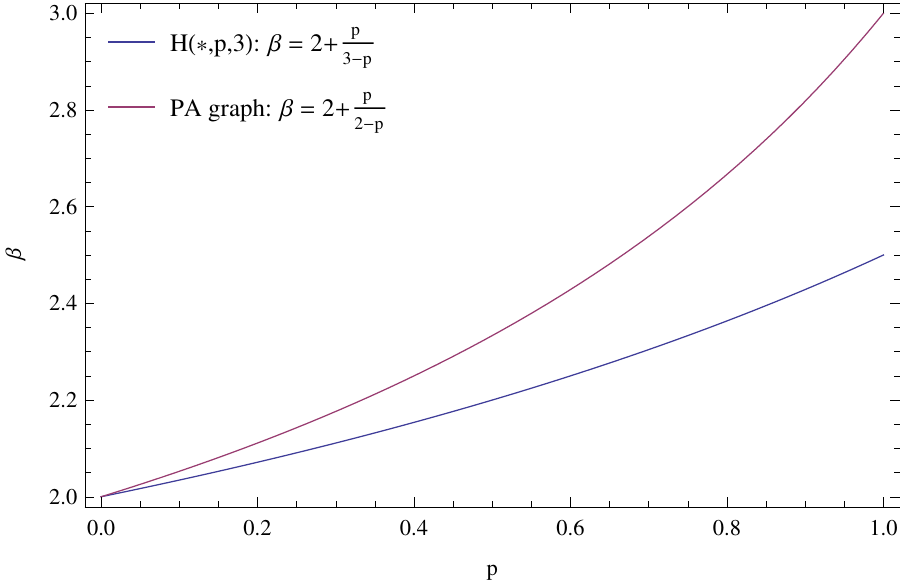}
\caption{The exponent $\beta$ of a preferential attachment graph and 
a 3-uniform hypergraph as a function of $p$ 
(the probability of an edge arrival event). In graphs it is between 2 and 3, 
whereas in 3-hypergraphs it is between 2 and 2.5.
}
\label{fig:exponent}
\end{figure}

Note that the expected degree of vertices in $G(H)$ in this case is $d(d-1)/2$. 
Interestingly, if we generate a new graph $G'$ with expected degree $d(d-1)/2$ 
according to the classical graph preferential attachment model, then its 
degree distribution will be $\beta'=2+p/(2-p)$. Hence the \emph{observed} 
degree distribution of $G(H)$ and $G'$, $\beta$ and $\beta'$ respectively, 
will be different. On the other hand, it we generate $G'$ 
(using the classical preferential attachment model) so that it agrees with 
the degree distribution of $G$, then the average degree will be different.  
This observation is supported by simulation results depicted in Figure 
\ref{fig:powerlaw}.
\begin{figure}[htb]
\centering
\includegraphics[width=.7\columnwidth]{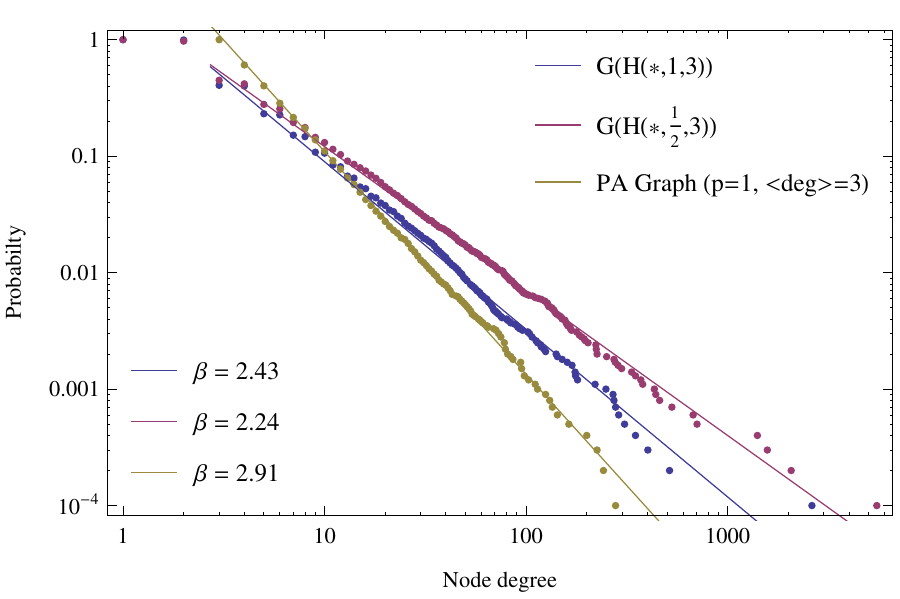}
\caption{Example of the cumulative degree distribution of three networks 
with $n=10,000$: 
(1) A graph $G(H(*,1,3))$ derived from a 3-uniform hypergraph $H(*,1,3)$, 
(2) A graph $G(H(*,\frac{1}{2},3))$ derived from a 3-uniform hypergraph 
$H(*,\frac{1}{2},3)$, and
(3) A preferential attachment graph with average degree $d(d-1)/2=3$. 
Graphs derived from hypergraphs have lower exponent, also as a function of $p$.}
\label{fig:powerlaw}
\end{figure}

This discussion seems to indicate that, in some sense, 
``the blanket (i.e., of the model) is too short" 
and one should be careful in deciding what is the right model that captures 
the observed degree distribution, and in particular, if the generative model 
is of a hypergraph or the classical graph model.

\section{Example}
\label{sec:dblp}

To test the above observations empirically, we studied a 
\emph{coauthorship hypergraph} of researchers in computer science, 
extracted from DBLP \cite{ley2009dblp}, a dataset recording most of the 
publications in computer science.
This hypergraph consists of hundreds of thousands of vertices 
(representing authors) and hyperedges (representing papers).
Figure \ref{fig:dblp-edges-power} shows the degree distribution 
of hyperedge sizes in DBLP for hyperedges sizes at least 3. 
The hyperedge size distribution closely fits a power law degree distribution 
with exponent $\beta=4.66$.  This means that the hyperedge size is
(with high probabilty) smaller than $m^{1/3}$, 
where $m$ is the number of papers (hyperedges). For the example of DBLP, 
where the number of papers is $m=2420879$, the number of authors on a paper 
(i.e., the hyper-edge size) will be with high probability below $134$.

\begin{figure}[h]
\centering
\includegraphics[width=.7\columnwidth]{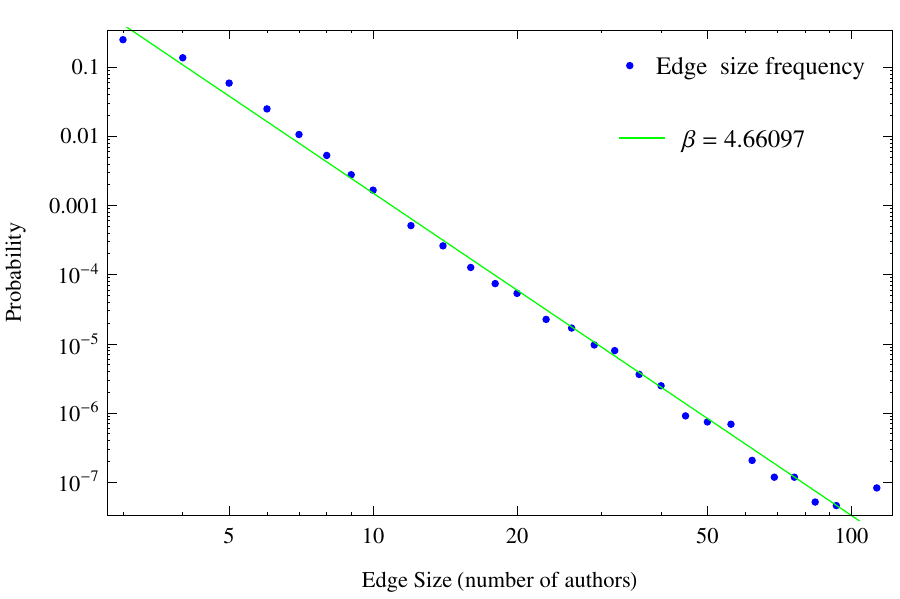}
\caption{The degree distribution of hyperedge sizes in DBLP 
for hyperedge sizes at least 3. 
The distribution closely fits a power law degree distribution 
with exponent $\beta=4.66$.}
\label{fig:dblp-edges-power}
\end{figure}

\section*{Acknowledgements}
The authors thank Eli Upfal for helpful discussions about the core ideas of the paper.

\clearpage

\bibliographystyle{acm}
\bibliography{hyperpa}

\end{document}